\documentclass[]{elsart}

\usepackage{amsmath}
\usepackage{amssymb}
\usepackage{url}
\usepackage{srcltx} 
\usepackage{enumerate}
\usepackage[utf8]{inputenc}
\usepackage{mathrsfs}
\usepackage[ruled,vlined]{algorithm2e}
\usepackage{xspace}

\def\floor#1{\left\lfloor #1 \right\rfloor}
\def\ceil#1{\left\lceil #1 \right\rceil}

\newcommand{\set}[1]{\{#1\}}

\newcommand{\ap}{\textsc{AvoidPattern}\xspace}

\newenvironment{proof}{\par \noindent \textbf{Proof.}
}{\hfill$\Box$\medskip}
\newenvironment{proofof}[1]{\par \noindent \textbf{Proof of #1.}
}{\hfill$\Box$\medskip} 

\newtheorem{theorem}{Theorem}
\newtheorem{lemma}[theorem]{Lemma}

\newcounter{example}
\date{\today}

\begin{document}

\begin{frontmatter}

  \title{Application of entropy compression\\ in pattern avoidance}
  
  \author{Pascal Ochem},
  \ead{Pascal.Ochem@lirmm.fr}
  \ead[url]{http://www.lirmm.fr/\~{}ochem}
  \author{Alexandre Pinlou\thanksref{UPV}}
  \ead{Alexandre.Pinlou@lirmm.fr}
  \ead[url]{http://www.lirmm.fr/\~{}pinlou}
  \thanks[UPV]{Second affiliation: Département de Mathématiques et
    Informatique Appliqués, 
    Université Paul-Valéry, Montpellier 3, Route de Mende, 34199
    Montpellier Cedex 5, France.} 
  
  \address{LIRMM - Univ. Montpellier 2, CNRS - 161 rue Ada, 34095
    Montpellier Cedex 5, France} 

\begin{abstract}
  In combinatorics on words, a word $w$ over an alphabet $\Sigma$ is
  said to avoid a pattern $p$ over an alphabet $\Delta$ if there is no
  factor $f$ of $w$ such that $f= h(p)$ where $h: \Delta^*\to\Sigma^*$
  is a non-erasing morphism. A pattern $p$ is said to be $k$-avoidable
  if there exists an infinite word over a $k$-letter alphabet that
  avoids $p$. We give a positive answer to Problem 3.3.2 in Lothaire's
  book ``Algebraic combinatorics on words'', that is, every pattern
  with $k$ variables of length at least $2^k$ (resp.
  $3\times2^{k-1}$) is 3-avoidable (resp. 2-avoidable).  This improves
  previous bounds due to Bell and Goh, and Rampersad.
\end{abstract}


\end{frontmatter}

\journal{}

%%%%%%%%%%%%%%%%%%%%%
\section{Introduction}\label{sec:intro}
%%%%%%%%%%%%%%%%%%%%%

A pattern $p$ is a non-empty word over an alphabet
$\Delta=\{A,B,C,\dots\}$ of capital letters called \emph{variables}.
A word $x$ over $\Sigma$ is an instance of $p$ if there exists a
non-erasing morphism $h: \Delta^*\to\Sigma^*$ such that $h(p)=x$. The
avoidability index $\lambda(p)$ of a pattern $p$ is the size of the
smallest alphabet $\Sigma$ such that there exists an infinite word $w$
over $\Sigma$ containing no instance of $p$ as a factor. Bean,
Ehrenfeucht, and McNulty~\cite{BEM79} and Zimin~\cite{Zimin}
characterized unavoidable patterns, i.e., such that
$\lambda(p)=\infty$. We say that a pattern $p$ is $t$-avoidable if
$\lambda(p)\le t$. For more informations on pattern avoidability, we
refer to Chapter 3 of Lothaire's book~\cite{Lothaire:2002}.

In this paper, we consider upper bounds on the avoidability index of
long enough patterns with $k$ variables. Bell and
Goh~\cite{BellGoh:2007} and Rampersad~\cite{Rampersad:2011} used a
method based on power series and obtained the following bounds:
\begin{theorem}[\cite{BellGoh:2007,Rampersad:2011}]
\label{BGR}
Let $p$ be a pattern with exactly $k$ variables.
\begin{itemize}
\item[(a)] If $p$ has length at least $2^k$ then
 $\lambda(p)\le4$.~\cite{BellGoh:2007}
\item[(b)] If $p$ has length at least $3^k$ then
 $\lambda(p)\le3$.~\cite{Rampersad:2011}
\item[(c)] If $p$ has length at least $4^k$ then
 $\lambda(p)=2$.~\cite{Rampersad:2011}
\end{itemize}
\end{theorem}

Our main result improves these bounds:
\begin{theorem}
\label{2tok}
Let $p$ be a pattern with exactly $k$ variables.
\begin{itemize}
\item[(a)] If $p$ has length at least $2^k$ then $\lambda(p)\le3$.
\item[(b)] If $p$ has length at least $3\times2^{k-1}$ then $\lambda(p)=2$.
\end{itemize}
\end{theorem}

Theorem~\ref{2tok} gives a positive answer to Problem 3.3.2 of
Lothaire's book~\cite{Lothaire:2002}. The bound $2^k$ in
Theorem~\ref{2tok}.(a) is tight in the sense that for every $k\ge 1$,
the pattern $p_k$ with $k$ variables in the family
$\{A,$ $ABA,$ $ABACABA,$ $ABACABADABACABA,\dots\}$ has length $2^k-1$ and is
unavoidable. Similarly, the bound $3\times2^{k-1}$ in
Theorem~\ref{2tok}.(b) is tight in the sense that
for every $k\ge 1$, the pattern with $k$ variables in the family
$\{AA,$ AABAA,$ AABAACAABAA,$ $AABAACAABAADAABAACAABAA,\dots\}$ has length
$3\times2^{k-1}-1$ and is not 2-avoidable. Hence, this
shows that the upper bound 3 of Theorem~\ref{2tok}.(a) is best possible.\\

The avoidability index of every pattern with at most 3 variables is
known, thanks to various results in the literature. In particular,
Theorem~\ref{2tok} is proved for $k\le 3$:
\begin{itemize}
\item For $k=1$, the famous results of Thue~\cite{Thue06,Thue12} give
$\lambda(AA)=3$ and $\lambda(AAA)=2$.
\item For $k=2$, every binary pattern of length at least 4 contains a
square, and is thus 3-avoidable. Moreover, Roth~\cite{Roth:1992}
proved that every binary pattern of length at least 6 is 2-avoidable.
\item For $k=3$, Cassaigne~\cite{Cassaigne:1994} began and the first
author~\cite{Ochem:2004} finished the determination of the
avoidability index of every pattern with at most 3 variables. Every
ternary pattern of length at least 8 is 3-avoidable and every binary
pattern of length at least 12 is 2-avoidable.
\end{itemize}

So, there remains to prove the cases $k\ge 4$.

%\medskip

Section~\ref{sec:preliminary} is devoted to some preliminary results.
We prove Theorem~\ref{2tok}.(a) in Section~\ref{3-a} as a corollary of
a result of Bell and Goh~\cite{BellGoh:2007}. In Section~\ref{2-a}, we
prove Theorem~\ref{2tok}.(b) using the so-called \emph{entropy compression method}.

%%%%%%%%%%%%%%%%%%%%%
\section{Preliminary results}\label{sec:preliminary}
%%%%%%%%%%%%%%%%%%%%%

Let $p$ be a pattern over $\Delta=\set{A,B,C,\ldots}$. An
\emph{occurrence} $y$ of $p$ is an assignation of a non-empty words
over $\Sigma$ to every variable of $p$ that form a factor. Note that
two distinct occurrences of $p$ may form the same factor. For example,
if $p=ABA$ then the occurrence $y=(A=00; B=1)$ of $p$ forms the factor
$00100$; on the other hand, $y'=(A=0; B=010)$ is a distinct occurrence
of $p$ which forms the same factor $00100$.

A pattern $p$ is \emph{doubled} if every variable of $p$ appears at least twice in $p$.

A pattern $p$ is \emph{balanced} if it is doubled and every variable
of $p$ appears both in the prefix and the suffix of length
$\floor{\frac{|p|}{2}}$ of $p$. Note that if the pattern has odd
length, then the variable $X$ that appears in the middle of $p$ (i.e. in
position $\floor{\frac{|p|}{2}} + 1$) must appear also in the prefix
and in the suffix in order to make $p$ balanced.

\begin{claim}\label{cl:balanced}
  For every integer $f\ge 2$, every pattern with at most $k$ variables and length
  at least $f\times 2^{k-1}$ contains a balanced pattern $p'$ with at most
  $k'\ge 1$ variables and length at least $f\times 2^{k'-1}$ as a factor.
\end{claim}

\begin{proof}
  We prove this claim by induction on $k$. If $k=1$, then $p$ has size
  at least $f\ge 2$ and is clearly balanced. Suppose this is true
  for some $k=n$, i.e. $p$ with $n$ variables and length at least
  $f\times 2^{n-1}$ contains a balanced pattern $p'$ as a factor
  with at most $k'$ variables and length at least $f\times
  2^{k'-1}$. Let $k=n+1$ and let $p_1$ (resp. $p_2$) be the prefix
  (resp. the suffix) of $p$ of size $\floor{\frac{|p|}{2}}$. If $p$
  is not balanced, then there exists a variable $X$ in $p$ that does
  not occur in $p_i$ for some $i\in\{1,2\}$. Thus, $p_i$ has 
  at most $k-1=n$ variables and length at least
  $f\times 2^{n-1}$. Therefore, by induction hypothesis, $p$
  contains a balanced pattern with at most $k'$ variables and
  length at least $f\times 2^{k'-1}$ as a factor.
\end{proof}

In the following, we will only use the fact that the pattern $p'$ in
Claim~\ref{cl:balanced} is doubled instead of balanced.

%%%%%%%%%%%%%%%%%%%%%
\section{3-avoidable long patterns}\label{3-a}
%%%%%%%%%%%%%%%%%%%%%
We prove Theorem~\ref{2tok}.(a) as a corollary of
the following result of Bell and Goh~\cite{BellGoh:2007}:
\begin{lemma}[\cite{BellGoh:2007}]
\label{k6doubled}
Every doubled pattern with at least 6 variables is 3-avoidable.
\end{lemma}

\begin{proofof}{Theorem~\ref{2tok}.(a)}
We want to prove that every pattern with exactly $k$ variables and
length at least $2^k$ is 3-avoidable, or equivalently, that every pattern
with at most $k$ variables and length at least $2^k$ is 3-avoidable.
By Claim~\ref{cl:balanced}, every such pattern contains a doubled pattern $p'$
as a factor with at most $k'\ge 1$ variables and length at least $2^{k'}$.
So there remains to show that every doubled pattern with at most $k$
variables and length at least $2^k$ is $3$-avoidable. As discussed in
the introduction, the case of patterns with at most $3$ variables has
been settled. Now, it is sufficient to prove that doubled patterns of
length at least $2^4=16$ are 3-avoidable.

Suppose that $p_1$ is a doubled pattern containing a variable $X$ that
appears at least 4 times. Replace $2$ occurrences of $X$ with a new
variable to obtain a pattern $p_2$. Example: We replace the first and
third occurrence of $B$ in $p_1=ABBCDBCABDDCB$ by a new variable $E$ to obtain
$p_2=AEBCDECABDDCB$. Then $p_2$ is a doubled pattern such that
$|p_1|=|p_2|$ and $\lambda(p_1)\le\lambda(p_2)$, since every
occurrence of $p_1$ is also an occurrence of $p_2$.

Given a doubled pattern $p$ of length at least $16$, we make such
replacements as long as we can. We thus obtain a doubled pattern $p'$
of length at least $16$ such that $\lambda(p)\le\lambda(p')$.
Moreover, every variable in $p'$ appears either $2$ or $3$ times and
therefore $p'$ contains at least $\ceil{16/3}=6$
variables. So $p'$ is $3$-avoidable by Lemma~\ref{k6doubled}. Thus $p$ is
$3$-avoidable, which finishes the proof.
\end{proofof}

%%%%%%%%%%%%%%%%%%%%%
\section{2-avoidable long patterns}\label{2-a}
%%%%%%%%%%%%%%%%%%%%%

We want to prove that every pattern with exactly $k$ variables and
length at least $3\times 2^{k-1}$ is 2-avoidable, or equivalently,
that every pattern with at most $k$ variables and length at least
$3\times 2^{k-1}$ is 2-avoidable. By Claim~\ref{cl:balanced}, every
such pattern contains a doubled pattern $p'$ as a factor with at most
$k'\ge 1$ variables and length at least $3\times 2^{k'-1}$. So there
remains to show that every doubled pattern with at most $k$ variables
and length at least $3\times 2^{k-1}$ is $2$-avoidable.

As discussed in the introduction, the case of patterns with at most
$3$ variables has been settled. Now, it is sufficient to prove
Theorem~\ref{2tok}.(b) for doubled patterns and $k\ge 4$. 

Let $\Sigma=\set{0,1}$ be the alphabet. For the remaining of this
section, let $k\ge 4$ and $q(k)=3\times 2^{k-1}$.

Suppose by contradiction that there exists a doubled pattern $p$ on $k$
variables and length at least $q(k)$ that is not $2$-avoidable.
Then there exists an integer $n$ such that any word $w\in\Sigma^n$ contains $p$.
We put an arbitrary order on the $k$ variables of $p$ and call $A_j$ the $j$-th variable of $p$.

%%%%%%%%%%%%%%%%%%%%%
\subsection{The algorithm \ap}
%%%%%%%%%%%%%%%%%%%%%

Let $V=\set{0,1}^t$ be a vector of length $t$. 
The following algorithm takes the vector $V$
as input and returns a word $w$ avoiding $p$ and a data structure $R$
that is called a \emph{record} in the remaining of the paper.

\begin{algorithm}[H]
\DontPrintSemicolon
\LinesNumbered

\SetKwInOut{Input}{Input}\SetKwInOut{Output}{Output}
\Input{$V$.}
\Output{$w$ (a word avoiding $p$) and $R$ (a data structure).\vspace{.2cm}}

  $w\gets \varepsilon$ 

  $R\gets \emptyset$
  
  \For{$i\gets 1$ \KwTo $t$}{
    Append $V[i]$ (the $i$-th letter of $V$) to $w$

    Encode in $R$ that a letter was appended to $w$
    
    \If{$w$ contains an occurrence of $p$ of length $\ell$}{
      Encode in $R$ the occurrence of $p$

      Erase the suffix of length $\ell$ of $w$
    }
  }
  \Return $R$, $w$
  \caption{\ap\label{algo:ap}}
\end{algorithm}

The way we encode information in $R$ at lines~5
and~7 will be explained in Subsection~\ref{subsec:record}.

In the algorithm \ap, let $w_i$ be the word $w$ after $i$ steps.
Clearly, $w_i$ avoids $p$ at each step. By contradiction hypothesis,
the resulting word $w$ of the algorithm (that is $w_t$) has length less than $n$.
We will prove that each output of the algorithm allows to determine the input.
Then we obtain a contradiction by showing that the number of possible outputs
is strictly smaller than the number of possible inputs
when $t$ is chosen large enough compared to $n$.
This implies that every pattern $p$ with at most $k$ variables and length at least
$q(k)$ is $2$-avoidable.

To analyze the algorithm, we borrow ideas from graph coloring
problems~\cite{Joret:2012,Esperet:2012}. These results are based on
the Moser-Tardos~\cite{moser:2010} entropy-compression method which
is an algorithmic proof of the Lov\'asz Local Lemma.

%%%%%%%%%%%%%%%%%%%%%
\subsection{The record $R$}\label{subsec:record}
%%%%%%%%%%%%%%%%%%%%%

An important part of the algorithm is to keep the record $R$ of
each step of the algorithm.
Let $R_i$ be the record after $i$ steps of the algorithm \ap. On one
hand, given $V$ as input of the algorithm, this produces a pair
$(R_t,w_t)$. On the other hand, given
a pair $(R_t,w_t)$, we will show in Lemma~\ref{lem:inj} that we can
recover the entire input vector $V$.
So, each input vector $V$ produces a distinct pair $(R_t,w_t)$.
 
Let $\mathcal{V}$ be the set of input vectors $V$ of size $t$, let
$\mathcal{R}$ be the set of records $R$ produced by the algorithm \ap
and let $\mathscr{O}$ be the set of different outputs $(R_t,w_t)$.
After the execution of the algorithm ($t$ steps), $w_t$ avoids $p$ by
definition and therefore $|w_t|<n$ by contradiction hypothesis.
Hence, the number of possible final words $w_t$ is independent from
$t$ (it is at most $2^n$). We then clearly have $|\mathscr{O}| \le
2^n\times |\mathcal{R}|$.  We will prove that $|\mathcal{V}|\le
|\mathscr{O}|$ and that $|\mathcal{R}|=o(2^t)$ to obtain the
contradiction $2^t=|\mathcal{V}|\le |\mathscr{O}| \le
2^n\times|\mathcal{R}|=o(2^t)$.

The record $R$ is a triplet $R=(D,L,X)$ where $D$ is a binary word
(each element is $0$ or $1$), $L$ is a vector of $(k-1)$-sets of
non-zero integers and $X$ is a vector of binary words.  At the
beginning, $D$, $L$ and $X$ are empty.  At step $i$ of the algorithm,
we append $V[i]$ to $w_{i-1}$ to get $w'_i$.

If $w'_i$ contains no occurrence of $p$, then we append $0$ to $D$ to
get $R_i$ and we set $w_i=w'_i$.  Otherwise, suppose that $w'_i$
contains an occurrence $y$ of $p$ that forms a factor $f$ of length
$\ell$ ($f$ is the $\ell$ last letters of $w'_i$). Recall that $A_j$
is the $j$-th variable of $p$.  Let $\ell_j=|A_j|$ in the factor $f$,
let $L_1 = \ell_1$, $L_j=L_{j-1}+\ell_j$ for $j\ge 2$.  Let
$L'=\set{L_1,L_2,\ldots,L_{k-1}}$ be a $(k-1)$-set of non-zero
integers.  Let $X'$ be the binary word obtained from $A_1\cdot
A_2\cdot\ldots\cdot A_k$ (where ``$\cdot$'' is the concatenation
operator) followed by as many $0$'s as necessary to get length
$\floor{\frac{\ell}{2}}$. Note that we necessarily have $|A_1\cdot
A_2\cdot\ldots\cdot A_k|\le \floor{\frac{\ell}{2}}$ since the pattern
is doubled.
Eventually, to get $R_i$,
we append the factor $01^\ell$ to $D$;
we add $L'$ as the last element of $L$; finally we add $X'$ as the last element of $X$.

\medskip

\textbf{Example:} Let us give an example with $k=3$, $p=ACBBCBBABCAB$
  and \linebreak
  $V=[0,0,1,0,0,1,1,0,0,1,1,1,0,0,1,1,0,1,1,1,0,0,0,1,1,0]$. The
  variables of $p$ were initially ordered as $(A,B,C)$. For the first
  $24$ steps, no occurrence of $p$ appeared, so at each step $i\le
  24$, we append $V[i]$ to $w_{i-1}$ and we append one $0$ to $D$.
  Hence, at step $24$, we have:
  \begin{itemize}
    \item $w_{24}=001001100111001101110001$ 
    \item $\displaystyle R_{24} =\left\{ 
        \begin{array}{rcl}
          D&=&000000000000000000000000=0^{24}\\ 
          L&=&[\ ]\\ 
          X&=&[\ ]
        \end{array} \right.$
  \end{itemize} 
  Now, at step $25$, we first append $V[25] = 1$ to $w_{24}$ to get
  $w'_{25}$. The word $w'_{25}$ contains an occurrence
  $y=(A=01;B=1;C=100)$ of $p$ which forms a factor of length $21$ (the
  $21$ last letters of $w'_{25}$). Then we set $L'=\set{2,3}$ and $X'=
  0111000000$ (this is $A\cdot B\cdot C$ followed by four $0$'s in order to
  get a binary word of length $10=\floor{\frac{21}{2}}$).
  Eventually, to get $w_{25}$ and $R_{25}$, we erase the suffix of length $21$
  of $w'_{25}$ to get $w_{25}$, we append the factor $01^{21}$ to $D$,
  $L'$ to $L$, and $X'$ to $X$. This gives:
  \begin{itemize}
    \item $w_{25}=0010$ 
    \item $\displaystyle R_{25} =\left\{ 
        \begin{array}{rcl}
          D&=&0000000000000000000000000111111111111111111111=0^{25}1^{21}\\ 
          L&=&[\set{2,3}] \\
          X&=&[0111000000]
        \end{array} \right.$
  \end{itemize}

This is where our example ends.

\medskip

Let $V_i$ be the vector $V$ restricted to its $i$ first elements.
We will show that the pair $(R_i,w_i)$ at some step $i$ allows to recover $V_i$.

\begin{lemma}\label{lem:inj}
  After $i$ steps of the algorithm \ap, the pair $(R_i,w_i)$ permits to recover $V_i$.
\end{lemma}

\begin{proof}
  Before step 1, we have $w_0=\varepsilon$, $R_0=(\varepsilon,[\ ],[\ 
  ])$, and $V_0=\varepsilon$.
  Let $R_i=(D,L,X)$ be the record after step $i$, with $1\le i\le t$. 
  \begin{itemize}
  \item Suppose that $0$ is a suffix of $D$. This means that at step
    $i$, no occurrence of $p$ was found: the algorithm appended $V[i]$
    to $w_{i-1}$ to get $w_i$. Therefore $V[i]$ is the last letter of
    $w_i$, say $x$.  Then the word $w_{i-1}$ is obtained from $w_i$ by
    erasing the last letter and the record $R_{i-1}$ is obtained from
    $R_i$ by removing the suffix $0$ of $D$. We recover $V_{i-1}$ from
    $(R_{i-1},w_{i-1})$ by induction hypothesis and we obtain
    $V_i=V_{i-1}\cdot x$.
  
  \item Suppose now that $01^\ell$ is a suffix of $D$. This means that
    an occurrence $y$ of $p$ which forms a factor $f$ of length $\ell$
    has been created during step $i$. The last element $L'$ of $L$ is
    a $(k-1)$-set $L'=\set{L_1,L_2,\ldots,L_{k-1}}$ and the last
    element $X'$ of $X$ is a binary word of length
    $\floor{\frac{\ell}{2}}$. Let $\ell_1=L_1$ and for $2\le s\le
    k-1$, let $\ell_s=L_{s}-L_{s-1}$. So, for $1\le s\le k-1$,
    $\ell_s$ is clearly the length of the variable $A_s$ of $p$ in the
    occurrence $y$ by construction of $L'$. We know the pattern $p$,
    the total length of the factor $f$ (that is $\ell$) and the
    lengths of the $k-1$ first variables of $p$ in $f$, so we are able
    to compute the length $\ell_k$ of the last variable $A_k$.  So we
    are now able to recover the occurrence $y$ of $p$: the first
    $\ell_1$ letters of $X'$ correspond to $A_1$, the next $\ell_2$
    letters correspond to $A_2$ and so on ($X'$ may contain some $0$'s
    at the end which are not relevant). It follows that the factor $f$
    is completely determined. So $w_{i-1}$ is obtained from $w_i\cdot
    f$ by removing the last letter $x$ of $f$, this letter $x$ being
    $V[i]$ (the appended letter to $w_{i-1}$ at step $i$ to get
    $w'_i$). The record $R_{i-1}$ is obtained from $R_i$ as follows:
    remove the suffix $01^\ell$ from $D$; remove the last element of
    $L$ and the last element of $X$.  We recover $V_{i-1}$ from
    $(R_{i-1},w_{i-1})$ by induction hypothesis and we obtain
    $V_i=V_{i-1}\cdot x$.
  \end{itemize}
\end{proof}

The previous lemma proves that distinct input vectors cannot
correspond to the same pair $(R_t,w_t)$. So we get $|\mathcal{V}|\le
|\mathscr{O}|$.

%%%%%%%%%%%%%%%%%%%%%
\subsection{Analysis of $\mathcal{R}$}
%%%%%%%%%%%%%%%%%%%%%

Now we compute $|\mathcal{R}|$. Let $R=R_t=(D,L,X)$ be a given record
produced by an execution of \ap. Let $\mathcal{D}$, $\mathcal{L}$ and
$\mathcal{X}$ be the set of such binary words $D$, of such
$(k-1)$-sets of non-zero integers $L$, and of such vectors of binary
words $X$, respectively. We thus have $|\mathcal{R}|\le
|\mathcal{D}|\times|\mathcal{L}|\times|\mathcal{X}|$.

Let us give some useful information in order to get upper bounds on
$|\mathcal{D}|$, $|\mathcal{X}|$, and $|\mathcal{L}|$. The algorithm
runs in $t$ steps.  At each step, one letter is appended to $w$, so
$t$ letters have been appended and therefore the number of erased
letters during the execution of the algorithm is $t-|w_t|$. At some
steps, an occurrence of $p$ appears and forms a factor which is
immediately erased. Let $m$ be the number of erased factors during the
execution of the algorithm.  Let $f_i$, $1\le i\le m$, be the $m$
erased factors.  We have $|f_i|\ge q(k)$ since each variable of $p$ is
a non-empty word and $p$ has length at least $q(k)$. Moreover, we have
$\sum_{1\le i\le m}|f_i| =t-|w_t|\le t$.  Each time a factor $f_i$ is
erased, we add an element to $L$ and $X$, so $|L|=|X|=m$.

%%%%%%%%%%%%%%%%%%%%%
\subsubsection{Analysis of $\mathcal{D}$}
%%%%%%%%%%%%%%%%%%%%%

In the binary word $D$, each $0$ corresponds to an
appended letter during the execution of the algorithm and each $1$
corresponds to an erased letter. Therefore, $D$ has length $2t-|w_t|$.
Observe that every prefix in $D$ contains at least as many $0$'s as $1$'s.
Indeed, since a $1$ corresponds to an erased letter $x$, this
letter $x$ had to be added first and thus there is a $0$ before that
corresponds to this $1$. The word $D$ is therefore a partial Dyck word.
Since any erased factor $f_i$ has length at least $q(k)$, any
maximal sequence of $1$'s (which is called a \emph{descent} in the
sequel) in $D$ has length at least $q(k)$. So $D$ is a partial Dyck
words with $t$ $0$'s such that each descent has length at least $q(k)$.
The following two lemmas due to Esperet and Parreau~\cite{Esperet:2012}
give an upper bound on $|\mathcal{D}|$.

Let $C_{t,r,d}$ (resp. $C_{t,d}$) be the number of partial Dyck words
with $t$ $0$'s and $t-r$ $1$'s (resp. Dyck words of length $2t$) such
that all descents have length at least $d$. Hence, we have $|\mathcal{D}|
\le C_{t,|w_t|,q(k)}$.

\begin{lemma}\cite{Esperet:2012}\label{louis1}
  $C_{t,r,d}\le C_{t+d-1,d}$.
\end{lemma}

Hence, we have  $|\mathcal{D}|\le C_{t+q(k)-1,q(k)}$.

Let $\phi_d(x) = 1+\sum_{i\ge d}x^i = 1+\frac{x^d}{1-x}$.

\begin{lemma}\cite{Esperet:2012}\label{louis2}
  Let $d$ be an integer such that the equation $\phi_d(x)-x\phi_d'(x)=0$
  has a solution $\tau$ with $0<\tau<r$, where $r$ is the radius of
  convergence of $\phi_d$. Then $\tau$ is the unique solution of the
  equation in the open interval $(0,r)$. Moreover, there exists a
  constant $c_d$ such that $C_{t,d}\le c_d\gamma_d^tt^{-\frac{3}{2}}$
  where $\gamma_d=\phi_d'(\tau)=\frac{\phi_d(\tau)}{\tau}$.
\end{lemma}

The solution of the equation $\phi_{d}(x)-x\phi_{d}'(x)=0$
is equivalent to $P(x)=(1-x)^2+(1-d)x^{d}+(d-2)x^{d+1}=0$.
The radius of convergence $r$ of $\phi_{d}$ is $1$ and since
$P(0)=1$ and $P(r)=-1$, $P(x)=0$ has a solution $\tau$ in the open
interval $(0,r)$. By Lemma~\ref{louis2}, this solution is unique and,
for some constant $c_{d}$, we have $C_{t+d-1,d}\le
c_{d}\gamma_d^{t+d-1}({t+d-1})^{-\frac{3}{2}}$ with
$\gamma_d=\phi'_{d}(\tau)$. We clearly have $C_{t+d-1,d}=o(\gamma_d^t)$.
So, we can compute $\gamma_d$ for $d$ fixed. We will use the following bounds:
$\gamma_{24}\le 1.27575$, $\gamma_{40}\le 1.15685$, and $\gamma_{100}\le 1.08603$.
Note that when $d$ increases, $\gamma_d$ decreases. 

So, by Lemmas~\ref{louis1} and~\ref{louis2}, when $t$ is large enough,
we have $|\mathcal{D}| < 1.27575^t$ (resp. $|\mathcal{D}| <
1.15685^t$, $|\mathcal{D}| < 1.08603^t$) if the length of any descent
is at least $24$ (resp. $48$, $100$).

%%%%%%%%%%%%%%%%%%%%%
\subsubsection{Analysis of $\mathcal{X}$}
%%%%%%%%%%%%%%%%%%%%%

Each element $X'$ of $X$ corresponds to an erased factor $f_i$ and by
construction $|X'| = \floor{\frac{|f_i|}{2}}$. So the sum of the
lengths of the elements of $X$ is $\floor{\frac{|f_1|}{2}} +
\floor{\frac{|f_2|}{2}} + \ldots + \floor{\frac{|f_m|}{2}}\le
\floor{\frac{t}{2}}$. Thus, the vector $X$ corresponds to a binary
word of length at most $\floor{\frac{t}{2}}$. Therefore
$|\mathcal{X}|\le 2^{\floor{\frac{t}{2}}}\le (\sqrt{2})^t$.

%%%%%%%%%%%%%%%%%%%%%
\subsubsection{Analysis of $\mathcal{L}$}
%%%%%%%%%%%%%%%%%%%%%

Each element $L'=\set{L_1,L_2,\ldots,L_{k-1}}$ of $L$ corresponds to
an erased factor $f_i$ and by construction each $L_j\in L'$ corresponds to the
sum of the lengths of the $j$ first variables of $p$ in $f_i$.

Let $h_k(\ell)$ be the number of such $(k-1)$-sets $L'$ that
correspond to factors of length $\ell$. Recall that $|f_i|\ge q(k)$,
so $h_k(\ell)$ is defined for $k\ge 4$ and $\ell\ge q(k)$.  Each of
the $m$ elements of $L$ corresponds to an erased factor, so
$|\mathcal{L}|\le h_k(|f_1|)\times h_k(|f_2|)\times\ldots\times
h_k(|f_m|)$.  Let $g_k(\ell)=h_k(\ell)^\frac{1}{\ell}$ defined for
$k\ge 4$ and $\ell\ge q(k)$.  Then $\mathcal{L}\le
g_k(|f_1|)^{|f_1|}\times g_k(|f_2|)^{|f_2|}\times \ldots\times
g_k(|f_m|)^{|f_m|}$. So, if we are able to upper-bound $g_k(\ell)$ by
some constant $c$ for all $\ell\ge q(k)$, then we get
$|\mathcal{L}|\le c^{|f_1|}\times c^{|f_2|}\times\ldots\times
c^{|f_m|}\le c^t$.

Now we bound $g_k(\ell)$ using two different methods
depending on the value of $k$ and the length $q(k)$ of $p$. 

\paragraph{Bound on $g_k(\ell)$ for $k=4$, $\ell\ge 100$\ \ or\ \ $k\ge 5$,
 $\ell\ge48$}\label{par:l1}

For any factor $f_i$, we have $|A_1\cdot A_2\cdot\ldots\cdot A_k|\le
\floor{|f_i|/2}$ since $p$ is doubled (each variable appears at least
twice). For a given $L'=\set{L_1,L_2,\ldots,L_{k-1}}$ that corresponds
to some factor $f_i$, we have $L_{k-1} = |A_1\cdot A_2\cdot\ldots\cdot
A_{k-1}|\le \floor{|f_i|/2}$. Therefore, $L'$ is a $(k-1)$-set of
distinct non-zero integers at most $\floor{|f_i|/2}$, i.e. $k-1$
integers chosen among integers between $1$ and $\floor{|f_i|/2}$; thus
$h_k(\ell)\le {\floor{\ell/2} \choose k-1}$ and so $g_k(\ell)\le
{\floor{\ell/2} \choose k-1}^{\frac{1}{\ell}}$.  We can upper-bound
$g_k(l)$ by $\overline{g_k}(l) =
\left(\frac{\left(\floor{\frac{\ell}{2}}\right)^{k-1}}{(k-1)!}\right)^{\frac{1}{\ell}}$
for $\ell\ge q(k)$. The function $\overline{g_k}(\ell)$ is decreasing
for $\ell\ge q(k)$; so $g_k(\ell)\le\overline{g_k}(\ell)\le
\overline{g_k}(q(k)) =
\left(\frac{\left(\floor{\frac{q(k)}{2}}\right)^{k-1}}{(k-1)!}\right)^{\frac{1}{q(k)}}$
for all $\ell\ge q(k)$. Moreover, we can see that
$\overline{g_{k}}(q(k))\le \overline{g_5}(48)$ for all $k\ge 5$.

Computing this upper bound, we get $g_k(\ell)\le
\overline{g_k}(q(k))\le \overline{g_5}(48) < 1.21973$ for all $k\ge 5$
and $\ell\ge 48$ and $g_4(\ell)\le \overline{g_4}(100) < 1.10456$ for
all $\ell\ge 100$.

\paragraph{Bound on $g_4(\ell)$ for $24\le \ell\le 99$} \label{par:l2}

The second method to bound the size of $g_4(\ell)$ is based on ordinary
generating functions (OGF). Here, $k=4$, so let $A_1,A_2,A_3,A_4$ be
the four variables of $p$ and let $a_i$ be the number of apparitions
of $A_i$ in $p$. Therefore, $a_1+a_2+a_3+a_4=|p|$. Recall that each
variable appears at least twice in $p$ since $p$ is doubled, so
$a_i\ge 2$. Moreover, a factor of length $\ell$, with $24\le\ell\le 99$, is
necessarily an occurrence of a pattern of length between $24$ and
$99$. So we just have to consider patterns $p$ with $24\le|p|\le 99$.

Given $L'=\set{L_1,L_2,L_3}$ an element of $L$ that corresponds to
some factor $f_i$, we can compute the lengths $\ell_i$ of each
variable $A_i$ in $f_i$ ($\ell_1 = L_1$, $\ell_i = L_i-L_{i-1}$ for
$i\in\set{2,3}$ and $\ell_4 =
\frac{|f_i|-(a_1\ell_1+a_2\ell_2+a_3\ell_3)}{a_4}$). Recall that
$\ell_i\ge 1$ since each variable of $p$ is a non-empty word. Let
$\mathcal{A}_{|p|}=\sum_{i\ge |p|}b_i\ x^i$ be the OGF of such sets
$L'$, i.e. $b_i$ is the number of $3$-sets $\set{L_1,L_2,L_3}$ that
corresponds to a factor of length $i$ formed by an occurrence of a
pattern of length $|p|$ (that is $b_i$ is the number of $4$-tuples
$(\ell_1,\ell_2,\ell_3,\ell_4)$ with $\ell_i\ge 1$ such that
$a_1\times \ell_1+ a_2\times \ell_2+ a_3\times \ell_3+a_4\times
\ell_4=i$). So by definition of $h_4$, we have $h_4(\ell)=b_\ell$.

This kind of OGF has been studied and is similar to the well-known
problem of counting the number of ways you can change a
dollar~\cite{polya:1983}: you have only five types of coins (pennies,
nickels, dimes, quarters, and half dollars) and you want to count the
number of ways you can change any amount of cents. So, let
$\mathcal{C}=\sum_{i\ge 1}c_i\ x^i$ be the OGF of the problem and thus
any $c_i$ is the number of ways you can change $i$ cents. Then, for example,
$c_{100}$ corresponds to the number of ways you can change a
dollar. Here, $\mathcal{C}=\frac{1}{1-x}\times \frac{1}{1-x^5}\times
\frac{1}{1-x^{10}}\times \frac{1}{1-x^{25}}\times \frac{1}{1-x^{50}}$.

In our case, we have four coins, each of them has value $a_i$ (so we
can have different types of coins with the same value) and each type
of coins appears at least once (since $\ell_i\ge 1$). Thus we get
$\mathcal{A}_{|p|}=\sum_{i\ge |p|}b_i\ x^i =
\frac{x^{a_1}}{1-x^{a_1}}\times\frac{x^{a_2}}{1-x^{a_2}}
\times\frac{x^{a_3}}{1-x^{a_3}}\times\frac{x^{a_4}}{1-x^{a_4}}$.
We use Maple for our computation. For each $24\le |p|\le 99$, for each
$4$-tuple $(a_1,a_2,a_3,a_4)$ such that $\sum a_i = |p|$, we consider
the associated OGF $\mathcal{A}_{|p|}$ and we compute, using Maple,
the truncated series expansion up to the order $100$, that gives
$\mathcal{A}_{|p|} = b_{24} x^{24} + b_{25} x^{25} + \ldots + b_{99}
x^{99} + O(x^{100})$ with explicit values for the coefficients $b_i$.
So, for any $24\le \ell\le 99$, $g_4(\ell)$ is upper-bounded by the
maximum of $(b_\ell)^{\frac{1}{\ell}}$ taken oven all
$\mathcal{A}_{|p|}$. Maple gives that $(b_\ell)^{\frac{1}{\ell}}$ is
maximal for $|p|=24$, $(a_1,a_2,a_3,a_4) = (2,2,2,18)$, and $\ell=46$:
in this case, $b_{46}=84$ (i.e. there exist $84$ distinct $3$-sets
$L'$ that correspond to some factor of length $46$ formed by an
occurrence of a pattern of length $24$). So, $g_4(\ell)\le
84^\frac{1}{46}< 1.10112$ for all $24\le \ell\le 99$, $k=4$ and
$24\le |p|\le 99$.

\paragraph{Bound $g(\ell)$ for all $k\ge 4$}

We can deduce from Paragraphs~\ref{par:l1} and~\ref{par:l2} the
following.

If $k=4$, then $g_4(\ell)< 1.10112$ for $24\le\ell\le 99$ and
$g_4(\ell)< 1.10456$ for $\ell\ge 100$. So for $k=4$, we have
$|\mathcal{L}|< (1.10456)^t$.

If $k\ge 5$, then $g_k(\ell)< 1.21973$ for $\ell\ge q(k)$.
So for $k\ge 5$, we have $|\mathcal{L}|< (1.21973)^t$.

%%%%%%%%%%%%%%%%%%%%%
\subsection{Conclusion}
%%%%%%%%%%%%%%%%%%%%%

Aggregating the above analysis, we get the following. For $k\ge 5$, we
have $q(k)\ge 48$: then $|\mathcal{R}|\le
|\mathcal{D}|\times|\mathcal{L}|\times|\mathcal{X}|\le (1.15685\times
1.21973\times\sqrt{2})^t =o(2^t)$. For $k=4$, we have $q(k)\ge 24$:
then
$|\mathcal{R}|\le|\mathcal{D}|\times|\mathcal{L}|\times|\mathcal{X}|
\le(1.27575\times1.10456\times\sqrt{2})^t = o(2^t)$.

Thus for all $k\ge 4$, $|\mathcal{R}| = o(2^t)$ and so we obtained the
desired contradiction: $$2^t=|\mathcal{V}|\le |\mathscr{O}| \le 2^n\times|\mathcal{R}|= 2^n
\times o(2^t)=o(2^t).$$

%%%%%%%%%%%%%%%%%%%%%
\section{Conclusion}
%%%%%%%%%%%%%%%%%%%%%

In our results, we heavily use the fact that the patterns are doubled.
The fact that the patterns are long is convenient for our proofs but
does not seem so important. So we ask whether every doubled pattern
is 3-avoidable. By the remarks in Section~\ref{sec:intro} and by
Lemma~\ref{k6doubled}, the only remaining cases are doubled patterns
with $4$ and $5$ variables. Also, does there exist a finite $k$
such that every doubled pattern with at least $k$ variables is
2-avoidable ?  We know that such a $k$ is at least 5 since ABCCBADD is
not 2-avoidable.

\end{document}